\documentclass[reqno]{amsart}
\calclayout

\usepackage{graphicx}
\usepackage{dcolumn}
\usepackage{bm}
\usepackage[utf8]{inputenc}
\usepackage[british]{babel}
\usepackage{amsfonts}
\usepackage{mathtools}
\usepackage{enumerate}
\usepackage{bbm}
\usepackage{xcolor}

\usepackage{mathdots}
\usepackage{MnSymbol}
\usepackage{tikz}
\usetikzlibrary{arrows,matrix,positioning}

\newcommand{\calgebra}{\mathfrak{A}}
\newcommand{\nalgebra}{\mathfrak{M}}

\newcommand{\hilbert}{\mathcal{H}}
\newcommand{\iu}{i\mkern1mu}
\newcommand{\Dom}[1]{\mathcal{D}\left(#1\right)}
\newcommand{\ie}{\textit{i}.\textit{e}.}
\newcommand{\eg}{\textit{e}.\textit{g}.}

\newcommand{\ip}[2]{\left\langle #1 , #2 \right \rangle}

\makeatletter
\newcommand*{\defeq}{\mathrel{\rlap{%
			\raisebox{0.3ex}{$\m@th\cdot$}}%
		\raisebox{-0.3ex}{$\m@th\cdot$}}%
	=}
\makeatother

\DeclareMathOperator{\Tr}{Tr}

\DeclareMathOperator{\Ran}{Ran}

\let\Im\relax
\newcommand{\Im}[1]{\operatorname{Im}\left(#1\right)}
\newcommand{\sgn}[1]{\operatorname{sgn}\left(#1\right)}
\usepackage{amsthm}
\theoremstyle{definition}
\newtheorem{definition}{Definition}[section]

\newtheorem{notation}[definition]{Notation}
\theoremstyle{plain}
\newtheorem{theorem}[definition]{Theorem}
\newtheorem{lemma}[definition]{Lemma}

\begin{document}

\title[]{Modular Structures on Trace Class Operators and Applications to Themodynamical Equilibrium States of Infinitely Degenerate Systems}

\author{R. Correa da Silva}
\email{ricardo.correa.silva@usp.br}
\address{Department of Mathematical Physics, Institute of Physics, University of S\~ao Paulo.}

	\begin{abstract}
	We study the thermal equilibrium states (KMS states) of infinitely degenerate Hamiltonians, in particular, we study the example of the Landau levels. We classify all KMS states in an example of algebra suitable for describing infinitely degenerate systems and we show that there is no cyclic and separating vector corresponding to the Landau Hamiltonian. Then, we try to reproduce the thermodynamical limit of a finite box as used in the very beginning of the theory of KMS states by Haag Hugenholtz and Winnink. Finally, we discuss the situation from the point of view of non-$\sigma$-additive probabilities, non-normal nor semifinite states, singular (Dixmier) states and, hence, an extension of the concept of KMS state.	\end{abstract}
	
	\maketitle

	\section{Introduction}
	\label{intro}
	
	 In this entire work we denote by $\calgebra$ a $C^\ast$-algebra and by $\nalgebra$ a von Neumann algebra. In addition, $\hilbert$ denotes a Hilbert space and often we will suppose $\calgebra, \nalgebra \subset B(\hilbert)$, \ie, the algebras will be thought as concrete ones.
	 
	\subsection{Landau Levels}\label{LandauLevels}
	
	The Landau quantization or Landau levels is the resulting quantum \hyphenation{de-scrip-tion}description of a simple physical problem that consists of a charged particle restricted to a plane with a constant uniform magnetic field in the orthogonal direction. The restriction to the plane plays no relevant role in our discussion and, considering that the particle is free in this direction, the momentum in $z$ and the Hamiltonian commutes and the only consequence is a continuous quantum number $p_z$ which adds a term $\frac{p_z^2}{2m}$ to the energy.
	
	The Hamiltonian of this problem can be obtained by means of the minimal coupling
	$$H=\frac{1}{2m}\left(p-\frac{q}{c}A\right)^2, \textrm{ where } \nabla \times A=B.$$
	
	In this work we will adopt a coordinate system in which the particle is in the $xy$ plane and $qB$ points in the direction of the $z$-axis. Also, we will adopt the symmetric gauge, namely $A=\frac{1}{2}\left(-Bx,By,0\right)$.
	
	Considering the dimensionless operators - which can be re-obtained in terms of $m$, $q$, $c$, $B$ and $\hbar$ - we can define the ladder operators
	
	\begin{equation}
	\begin{aligned}
	\label{ladderop}
	&a&=\frac{1}{\sqrt{2}}\left[\left(\frac{x}{2}+\frac{\partial}{\partial x}\right)-\iu \left(\frac{y}{2}+\frac{\partial}{\partial y}\right)\right]\\
	&a^\dagger&=\frac{1}{\sqrt{2}}\left[\left(\frac{x}{2}-\frac{\partial}{\partial x}\right)+\iu \left(\frac{y}{2}-\frac{\partial}{\partial y}\right)\right]\\
	&b&=\frac{1}{\sqrt{2}}\left[\left(\frac{x}{2}+\frac{\partial}{\partial x}\right)+\iu \left(\frac{y}{2}+\frac{\partial}{\partial y}\right)\right]\\
	&b^\dagger&=\frac{1}{\sqrt{2}}\left[\left(\frac{x}{2}-\frac{\partial}{\partial x}\right)-\iu \left(\frac{y}{2}-\frac{\partial}{\partial y}\right)\right],\\
	\end{aligned}
	\end{equation}
	such that $[a,a^\dagger]=[b,b^\dagger]=1$ and $[a,b]=[a,b^\dagger]=0$. For these operators, the dimensionless Hamiltonian and angular momentum in the $z$-axis are
	$H=a^\dagger a +\frac{1}{2}$ and $L_z=b^\dagger b+a^\dagger a$, respectively.
	
	Since the Hamiltonian and the angular momentum commute, we can choose a simultaneous basis of eigenvectors denoted $|n,m\rangle$ for $n\geq 0$ and $m\geq-n$, such that 
	\begin{equation}
	\label{eigenvalues}
	H|n,m\rangle=\left(n+\frac{1}{2}\right)|n,m\rangle \ \textrm{ and } \ L_z|n,m\rangle=-m|n,m\rangle.
	\end{equation}
	
	The ladder operators acting on these eigenstates satisfy
	\begin{equation}
	\begin{aligned}
	\label{laddereigen}
	&a|n,m\rangle=\sqrt{n}|n-1,m+1\rangle, \quad \textrm{ for } n\geq 1\\
	&a|0,m\rangle=0\\	&b|n,m\rangle=\sqrt{m+n}|n,m+1\rangle \quad \textrm{ for } m\geq -n+1\\
	&b|n,-n\rangle=0.\\
	\end{aligned}
	\end{equation}
	
	\subsection{Modular Theory}
	\label{subsecModular}
	Let us now define two operators in $\nalgebra$, which give rise to the operators that give name to this section. For the cyclic and separating vector $\Omega$, define the antilinear operators:
	\begin{center}
		\begin{minipage}[c]{0.44\textwidth}
			\[\begin{aligned}
			S_0 :	&\{A\Omega \in \hilbert \ | \ A\in \nalgebra\} 		&\to		& \hspace{0.4cm} \hilbert\\
			&\hspace{1.2cm} A\Omega 								&\mapsto	&\hspace{0.2cm} A^\ast\Omega
			\end{aligned}\]
		\end{minipage}\hspace{0.03\textwidth}and\hspace{0.03\textwidth}
		\begin{minipage}[c]{0.44\textwidth}
			\[\begin{aligned}
			F_0 :	& \{A^\prime\Omega \in \hilbert \ | \ A^\prime\in \nalgebra^\prime\} 		&\to		& \hspace{0.4cm} \hilbert \\
			& \hspace{1.2cm} A^\prime\Omega 								&\mapsto	&\hspace{0.2cm}A^{\prime\ast}\Omega
			\end{aligned}\]
		\end{minipage}
	\end{center}
	Note that the domains of the operators are dense subspaces.
	It is a standard result that the operators $S_0$ and $F_0$ are closable operators. Moreover, $S_0^\ast=\overline{F_0}$ and $F_0^\ast=\overline{S_0}$.	We will denote $S=\overline{S_0}$ and $F=\overline{F_0}$. An important point to stress now is that we omitted the dependence on $\Omega$ to keep the notation clean, but we will mention it in the following. Moreover, even though $S$ is not a bijection, it is injective and we write $S^{-1}$ (which is equal to $S$) to denote its inverse over its range. The same holds for $\Delta_\Omega$, which will be defined soon.
	
	\begin{definition}
		We denote by $J_\Omega$ and $\Delta_\Omega$ the unique anti-linear partial isometry and positive operator, respectively, in the polar decomposition of $S$, \ie,   $S=J_\Omega\Delta_\Omega^{\frac{1}{2}}$. $J_\Omega$ is called the modular conjugation and $\Delta_\Omega$ is called the modular operator. 
	\end{definition}
	
	Note that the existence and uniqueness of these operators are stated in the Polar Decomposition Theorem.
	
	Several properties hold for the modular operator, we refer to \cite{Bratteli1}, \cite{Araki74}, and \cite{RCS18} for the reader interested in this subject.
	
	One of the most important results in Modular Theory is the Tomita-Takesaki Theorem, which is extremely significant to both Physics and Mathematics. The proof and applications of this theorem can be found in \cite{Bratteli1}
	and \cite{Takesaki2002}. One of the consequences of this theorem is that, for each fixed $t\in\mathbb{R}$, $A\xrightarrow{\tau^\Omega_t}\Delta_\Omega^{\iu t}A\Delta_\Omega^{-\iu t}$ defines an isometry of the algebra. Hence, $\left\{\tau^\Omega_t \right\}_{t\in\mathbb{R}}$ is a one-parameter group of isometries.
	
	\begin{definition}[Modular Automorphism Group]
		Let $\nalgebra$ be a von Neumann algebra with cyclic and separating vector $\Omega$ and let $\Delta_\Omega$ be the associated modular operator. For each $t\in\mathbb{R}$, define the isometry $\tau^\Omega:\nalgebra \to \nalgebra$ by $\tau^\Omega_t(A)=\Delta_\Omega^{\iu t}A\Delta_\Omega^{-\iu t}$. We call the modular automorphism group\index{modular! automorphism group} the one-parameter group $\left\{\tau^\Omega_t \right\}_{t\in\mathbb{R}}$.
	\end{definition}
	
	\begin{notation}
		We will denote the modular automorphism group with respect to a cyclic and separating vector $\Omega$ by $\left\{\tau^\Omega_t \right\}_{t\in\mathbb{R}}$. In addition, given a faithful normal semifinite weight $\phi$ on a von Neumann algebra, we will denote $\left\{\tau^\phi_t \right\}_{t\in\mathbb{R}}$ the modular automorphism group with respect to the cyclic and separating vector obtained in the GNS-construction.
	\end{notation}
	
	All the topics mentioned here can be found in detail in any classical book on the subject, \eg \, \cite{Bratteli1}, \cite{KR83}, and \cite{Takesaki2002} or even in \cite{Araki74}.

	\subsection{KMS States and Dynamical Systems}
	\label{subsecKMS}

	We denote by $(\nalgebra,\alpha)$ the $W^\ast$-dynamical system with $\alpha$ a one-parameter group, $\mathbb{R} \ni t \mapsto \alpha_t \in Aut(\nalgebra)$. A general definition of KMS states can be found in any textbook on Operator Algebras such as \cite{Takesaki2002}, \cite{Bratteli2}, and \cite{KR86}. For a discussion in the context of equilibrium states in the thermodynamic limit we suggest reference \cite{haag67}. We will present this definition for completeness.
	
	\begin{definition}
		\label{KMSequi}
		Let $(\nalgebra, \tau)$ be a $W^\ast$-dynamical system and $\beta \in \mathbb{R}$. A normal state $\omega$ over $\nalgebra$ is said to be a $(\tau,\beta)$-KMS state, if, for any $A,B \in \nalgebra$, there exists a complex function $F_{A,B}$ which is analytic in $\mathcal{D}_\beta=\left\{z\in \mathbb{C} \mid 0<\sgn{\beta} \Im{ z}<|\beta|\right\}$ and continuous on $\overline{\mathcal{D}_\beta}$ satisfying
		\begin{equation}
		\label{eqKMS3}
		\begin{aligned}
		F_{A,B}(t) &=& \omega(A\tau_t(B)) \ \forall t\in \mathbb{R}, \\
		F_{A,B}(t+\iu \beta)& =& \omega(\tau_t(B)A) \ \forall t\in \mathbb{R}. \\
		\end{aligned}
		\end{equation}
		
	\end{definition}
	
We finish this section mentioning a result of the most importance. KMS states ``survive the thermodynamical limit'', which means that, under the topology that has physical significance in this mathematical description and under the adequate convergence and continuity hypothesis, the limit of KMS states is again a KMS state.
	
There is an important connection between KMS states and the modular operator, since the modular operator is a KMS state for the inverse temperature $\beta=-1$. The reason for the negative sign relies on a difference in convention among physicists and mathematicians which also reflects on the sign of the Hamiltonian. In the present work we will prefer the mathematicians convention unless explicitly stated, namely
$\alpha_t(A)=e^{iHt}Ae^{-iHt}$.
	
\subsection{Two von Neumann Algebras}\label{TvNA}
	
	Henceforth, consider $\hilbert$ to be a separable Hilbert space and $\{e_i\}_{i\in \mathbb{N}}$ an orthonormal basis of $B(\hilbert)$. Then, it is well known that
	$$\Tr(A)=\sum_{i=1}^\infty \ip{Ae_i}{e_i}$$ defines a normal faithful semifinite trace in $B(\hilbert)$.
	
	We know that the set of Hilbert-Schmidt operators $HS(\hilbert)\defeq \left\{A\in B(\hilbert) \ \middle|\ \Tr \left(|A|^2\right)<\infty \right\}$ is also a Hilbert space, where its inner product is given by
	$\ip{A}{B}_{HS}=\Tr(AB^\ast)$.
	
	In the following, we denote by $e\otimes f: \hilbert \to \hilbert$ the linear operator defined by $(e\otimes f)x=\ip{e}{x}f$, where $e,f\in \hilbert$. Notice that
	$\Tr(|e\otimes f|^2)=\|e\| \|f\|$, thus $e\otimes f\in HS(\hilbert)$. In addition, notice that $\{e_i\otimes e_j\}_{i,j\in\mathbb{N}}$ is an orthonormal basis of $HS(\hilbert)$. Hence, each $A\in HS(\hilbert)$ can be uniquely written as $A=\sum_{i,j=1}^\infty a_{ij} \ e_i\otimes e_j$, with $\|A\|_{HS}^2=\sum_{i,j=1}^\infty |a_{ij}|^2$. Since the converse is also true, there exists an isomorphism of Banach spaces between $HS(\hilbert)$ and the set $\ell_2(\mathbb{N}^2,\mathbb{C})=\left\{(a_{ij})_{i,j\in\mathbb{N}}\subset \mathbb{C}\ \middle| \ \sum_{i,j=1}^\infty |a_{ij}|^2<\infty\right\}$ with the norm $\left\| (a_{ij})_{i,j\in\mathbb{N}}\right\|^2=\sum_{i,j=1}^\infty |a_{ij}|^2$, which can be shown to be isomorphic as a Banach space to $\ell_2(\mathbb{C})$ by just using an enumeration of $\mathbb{N}^2$.

	We are not going to present a review about Schatten classes, but it is important to just mention two trace inequalities that will be used in this text and whose proofs can found in \cite{RCS18} and the references therein.
	\begin{lemma}
		\label{normtraceinequality}
		Let $\nalgebra$ be a von Neumann algebra, $\tau$ a normal faithful semifinite trace on $\nalgebra$, $A\in \nalgebra$ and $B\in \mathfrak{M}_\tau$. Then
		$\left|\tau(AB)\right|\leq\tau(|AB|)\leq \|A\|\tau(|B|)$.
	\end{lemma}
	
	\begin{theorem}[Minkowski's Inequality\index{inequality! Minkowski}]
		\label{minkowski}
		Let $\nalgebra$ be a von Neumann algebra, $\tau$ a normal faithful semifinite trace in $\nalgebra$, and $p,\, q>1$ such that $\frac{1}{p}+\frac{1}{q}=1$. Then
		\begin{enumerate}[(i)]
			\item for every $A\in \nalgebra$, $\displaystyle \tau(|A|^p)^\frac{1}{p}=\sup\left\{\left|\tau(AB)\right| \ \middle | \ B\in \nalgebra,  \tau\left(|B|^q\right)\leq 1\right\};$
			\item for every $A,B \in \nalgebra$, $\displaystyle \|A+B\|_p\leq \|A\|_p+\|B\|_p$.
		\end{enumerate}
	\end{theorem}
	
	Here, two von Neumann algebras that seem suitable for the description of systems with infinite degeneracy, such as the Landau levels problem are presented. These von Neumann algebras are known, but the author is not aware of any reference containing the proofs of this known facts, hence we will present the author's proofs for completeness.
	
	Let us define some operators of interest in $B(HS(\hilbert))$. Since the set of Hilbert-Schimidt operators (in fact, all Schatten classes) is a $\ast$-ideals of compact operators, given  $A,B\in HS(\hilbert)$ we can define the operator $A\vee B: HS(\hilbert)\to HS(\hilbert)$ by $(A\vee B)X=AXB^\ast$.
	By Lemma \ref{normtraceinequality} we have that $$\|(A\vee B)X\|_{HS}=\|AXB^\ast\|_{HS}=\left(\Tr\left(|AXB^\ast|^2\right)\right)^{\frac{1}{2}}\leq \|A\| \|B\|\|X\|_{HS},$$
	hence $A\vee B\in B(HS(\hilbert))$.
	Furthermore,
	$$(A_1\vee B_1)(A_2\vee B_2)X=(A_1\vee B_1)A_2 X B_2^\ast=A_1 A_2 X B_2^\ast B_1^\ast=(A_1 A_2 \vee B_1 B_2)X, \ \forall X \in HS(\hilbert)$$ and
	$$\begin{aligned}
	\ip{(A\vee B)X}{Y}_{HS}&=\Tr(A X B^\ast Y^\ast)\\
	&=\Tr(X B^\ast Y^\ast A)\\
	&=\Tr(X (A^\ast Y B)^\ast)\\
	&=\ip{X}{(A^\ast \vee B^\ast)Y}, \forall X, Y \in HS(\hilbert).
	\end{aligned}$$
	Hence, $(A_1\vee B_1)(A_2\vee B_2)=(A_1 A_2\vee B_1 B_2)$ and $(A\vee B)^\ast = (A^\ast\vee B^\ast)$.
	
	Now it is easy to see that $$\calgebra_l=\{A\vee \mathbbm{1} \in B(HS(\hilbert)) \ | \ A\in HS(\hilbert)\} \textrm{ and } \calgebra_r=\{\mathbbm{1}\vee B \in B(HS(\hilbert)) \ | \ B\in HS(\hilbert)\}$$ are two $\ast$-subalgebras of $B(HS(\hilbert))$.
	
	Let us show that $(\calgebra_r)^\prime=\calgebra_l$. Let $H\in (\calgebra_l)^\prime$, then $(\mathbbm{1}\vee A)H(X)=H(X)A^\ast=H(XA^\ast)=H(\mathbbm{1}\vee A)X$ for every Hilbert-Schmidt operator $A$.
	
	Define the operator $\tilde{H}:\hilbert \to \hilbert$ by $\tilde{H}y=H(P_y)y$, where $P_y$ is the orthogonal projection on $span[\{y\}]$. Notice now that, if $X$ if a finite-rank operator, $\overline{\Ran{X}}$ must be finite dimensional. Let $R_X$ denote the orthogonal projection on $\overline{\Ran{X}}$. Then, for every $y\in\hilbert$ an any $P$ projection with $P_{Xy}\leq P$,	
	\begin{equation}
	\label{calcx2}
	\begin{aligned}
	H(X)y&=H(X P_y) y\\
	&=H(R_{X P_y})X y\\
	&=H(P_{X y})X y\\
	&=H(P P_{X y}) X y\\
	&=H(P)P_{X y} X y\\
	&=H(P) X y.
	\end{aligned}
	\end{equation}
	It follows from the equation \eqref{calcx2} that $\tilde{H}$ is a linear operator, since $$\begin{aligned}
	\tilde{H}(x+y)&=H(P_{x+y})(x+y)\\
	&=H(P_{x+y})x+H(P_{x+y})y\\
	&=H(P)x+H(P)_y\\
	&=H(P_x)x+H(P_y)y\\
	&=\tilde{H}x+\tilde{H}y,
	\end{aligned}$$
	where $P$ is the orthogonal projection on $span[\{x,y\}]$. Now, it follows easily that $\tilde{H}$ is a bounded operator and, again by equation \eqref{calcx2}, it follows that $(\tilde{H}\vee \mathbbm{1})Xy=\tilde{H}Xy=H(X)y$, which means that $H(X)=(\tilde{H}\vee\mathbbm{1})X$. Since $H$ and $\tilde{H}\vee\mathbbm{1}$ coincide in finite-rank operators and these operators constitute a dense subset of the set of compact operators, it follows that $H=\tilde{H}\vee\mathbbm{1}$ as operators in $B(HS(\hilbert))$.
	
	Finally, the conclusion is that, if $H\in(\calgebra_r)^\prime$, there exists $\tilde{H}\in B(\hilbert)$ such that $H=\tilde{H}\vee \mathbbm{1}$ and, since the converse is quite easy, we conclude that $(\calgebra_r)^\prime=\calgebra_l$.
	
	The proof that $(\calgebra_l)^\prime=\calgebra_r$ follows by a similar argument, but using adjoints.
	
	Notice that the von Neumann Bicommutant Theorem implies that $\calgebra_l$ and $\calgebra_r$ are von Neumann algebras.

	\section{Classification of KMS States on a von Neumann Algebra for Infinitely Degenerate Systems} \label{classification}
	
	In this section we present the main results of this work. The von Neumann algebras presented in Subsection \ref{TvNA} are suitable for the treatment of infinitely degenerate Hamiltonians, as is the case of the Landau Hamiltonian that is discussed in \cite{bagarello2010}. Hence, from the point of view of Tomita-Takesaki modular theory,  characterizing all cyclic and separating vectors for this algebras can enlighten not only the modular structure of these algebras, but also provide relevant information about the possible faithful (normal) KMS states on these algebras. That is exactly what is done in Subsection \ref{MS}; then, in Subsection \ref{H} we use Modular Theory and the results of the previous subsection to obtain the Hamiltonians for which there are KMS states. In particular, we check that there is no KMS state to the Hamiltonian that originates the famous Landau levels, which contrasts with previous results \cite{bagarello2010}; since it is well known that restricting the system to a finite box ceases the infinite degeneracy of the Landau levels, we analyse what happens with the thermodynamical limit of the Gibbs equilibrium states in a box.
	
	\subsection{Modular Structure }\label{MS}

In order to construct the modular operators, a cyclic and separating vector is required. Since our algebras of interest are $\nalgebra_r$ and $\nalgebra_l$, the main objective here is to study necessary and sufficient conditions for a vector of $HS(\hilbert)$ to be cyclic and separating for these algebras.
	
	\begin{theorem}
		\label{eqciclicseparating}
		Let $\hilbert$ be a separable infinite dimensional Hilbert space, $\{e_n\}_{n\in\mathbb{N}}\subset \hilbert$ be an orthonormal basis of $\hilbert$ and $\Omega=\sum_{m=1}^\infty\sum_{n=1}^\infty a_{mn}e_m\otimes e_n\in HS(\hilbert)$. The following statements are equivalent:
		
		\begin{enumerate}[(i)]
			\item $\Omega$ is separating for $\calgebra_l$ and $\calgebra_r$; 
			\item $\displaystyle \left(\sum_{n=1}^\infty a_{mn}e_n\right)_{m\in\mathbb{N}}$ is a basis of $\hilbert$;
			\item $\displaystyle \left(\sum_{m=1}^\infty a_{mn}e_m\right)_{n\in\mathbb{N}}$ is a basis of $\hilbert$;
			\item $\displaystyle \overline{span\left\{\sum_{n=1}^\infty a_{mn}e_n\right\}_{m\in\mathbb{N}}}=\hilbert=\overline{span\left\{\sum_{m=1}^\infty a_{mn}^\ast e_m\right\}_{n\in\mathbb{N}}}$.
		\end{enumerate}
	\end{theorem}
	\begin{proof}
		In order do simplify the notation, let us set $y_m=\sum_{n=1}^\infty a_{mn}e_n$ and $z_n=\sum_{m=1}^\infty a_{mn}^\ast e_m$ in the following. With these definitions we have $\Omega =\sum_{m=1}^\infty e_m\otimes y_m=\sum_{n=1}^\infty z_n\otimes e_n$.
		
		$(ii)\Rightarrow (i)$ If $(y_m)_{m\in\mathbb{N}}$ is a basis of $\hilbert$,
		\begin{equation}\begin{aligned}
		\label{calculation1}
		& \ 0=(A\vee\mathbbm{1})\Omega=A\Omega \\
		\Leftrightarrow & \ 0=A\Omega x =\sum_{m=1}^\infty \ip{e_m}{x} Ay_m, \quad \forall x \in \hilbert \\
		\Leftrightarrow & \ Ay_k=0, \quad \forall k\in\mathbb{N}\\
		\Leftrightarrow & \ A=0,
		\end{aligned}\end{equation}
		and
		\begin{equation}\begin{aligned}
		\label{calculation2}
		& \ 0=(\mathbbm{1}\vee A)\Omega=\Omega A^\ast \\
		\Leftrightarrow & \ 0=\Omega A^\ast x =\sum_{m=1}^\infty \ip{e_m}{A^\ast x} y_m, \quad \forall x \in \hilbert \\
		\Leftrightarrow & \ 0=\ip{A e_n}{x}, \quad \forall n\in\mathbb{N}, \ \forall x \in \hilbert\\
		\Leftrightarrow & \ A e_k=0, \quad \forall k \in \mathbb{N}\\
		\Leftrightarrow & \ A=0.
		\end{aligned}\end{equation}

		$(i)\Rightarrow (ii)$ We will prove the contrapositive.
		
		Suppose first that $\overline{span[\{y_m\}_{m\in\mathbb{N}}]}\neq\hilbert$. Then, there exists $y\in\overline{span[\{y_m\}_{m\in\mathbb{N}}]}^\perp$ and $A=y\otimes e_1 \in HS(\hilbert)\setminus\{0\}$ is such that $(A\vee\mathbbm{1})\Omega=A\Omega=0$, since $\Ran(\Omega)=\overline{span[\{y_m\}_{m\in\mathbb{N}}]}$.
		
		Suppose now that there exists some $k\in\mathbb{N}$ such that $y_k=\sum_{j=1}^{k-1} \lambda_j e_j$. Then 
		$$
		\begin{aligned}
		\Omega&=\sum_{m=1}^\infty e_m\otimes y_m\\
		&= \sum_{\substack{m=1 \\ m\neq k}}^\infty e_m\otimes y_m+\sum_{j=1}^{k-1}\lambda_j e_k\otimes y_j\\
		&=\sum_{m=1}^{k-1} \left(e_m+\lambda_m^\ast  e_k\right)\otimes y_m+\sum_{m=k+1}^{\infty} e_m\otimes y_m.
		\end{aligned}$$
		
		Notice that $F\defeq\overline{span\left[\{e_m+\lambda_m^\ast e_k\}_{m=1}^{k-1}\cup\{e_m\}_{m=k+1}^\infty\right]}=\overline{span\left[\{e_m\}_{m\in\mathbb{N}}\setminus\{e_k\}\right]}\neq \hilbert$, hence there exists $w\in F^\perp\setminus\{0\}$. Then, $A^\ast=e_1\otimes w \neq 0$ is such that $(\mathbbm{1}\vee A)\Omega=\Omega A^\ast=0$ as can be checked in equation \eqref{calculation2}.
		
		From what we have now, we can conclude that if we perform the Gram-Schmidt process in the sequence $(y_m)_{m\in\mathbb{N}}$ it will never vanish. Thus, we will end up with an orthonormal sequence $(\tilde{y}_m)_{m\in\mathbb{N}}$ such that $\overline{ span\left[\{\tilde{y}_m\}_{m\in\mathbb{N}}\right]}=\overline{ span\left[\{y_m\}_{m\in\mathbb{N}}\right]}=\hilbert$. In other words, $(\tilde{y}_m)_{m\in\mathbb{N}}$ is an orthonormal basis of $\hilbert$. Since the coefficients of any vector are unique in the basis $(\tilde{y}_m)_{m\in\mathbb{N}}$ and each $y_k\in span\left[\{\tilde{y}_m\}_{m=1}^{k-1}\right]$, it follows that $(y_m)_{m\in\mathbb{N}}$ is a Schauder basis for $\hilbert$. 
		
		$(ii) \Leftrightarrow (iii)$ Notice that $(A\vee \mathbbm{1})\Omega=A\Omega=0 \Leftrightarrow (\mathbbm{1}\vee A)\Omega^\ast=\Omega^\ast A^\ast=0$, $(\mathbbm{1}\vee A)\Omega=\Omega A^\ast=0 \Leftrightarrow (A^\ast \vee \mathbbm{1})\Omega^\ast=A^\ast \Omega^\ast=0$, and $\Omega^\ast$ interchange the roles of $(y_m)_{m\in\mathbb{N}}$ and $(z_n)_{n\in\mathbb{N}}$
		
		$(i)\Leftrightarrow (iv)$ It is the exact same argument used in the beginning of $(ii)\Leftarrow (i)$ to show that $\overline{span[\{y_n\}_{n\in\mathbb{N}}]}=\hilbert$ and an analogous argument to prove that $\overline{span[\{z_n\}_{n\in\mathbb{N}}]}=\hilbert$
	\end{proof}
	
	Let $\Omega=\sum_{m=1}^\infty\sum_{n=1}^\infty a_{mn}\ e_m\otimes e_n=\sum_{m=1}^\infty e_m\otimes y_m\in HS(\hilbert)$ be a cyclic and separating vector for $\calgebra_l$ and let us continue denoting $y_m=\sum_{n=1}^\infty a_{mn}e_n$ and $z_n=\sum_{m=1}^\infty a_{mn}^\ast e_m$. The antilinear operators $S_0: \calgebra_r\Omega \to \hilbert$ and $F_0: \calgebra_l\Omega \to \hilbert$ used in the definition of the modular operators are given by 
	$$\begin{aligned}
	S_0((A\vee \mathbbm{1})\Omega)&=S_0\left(\sum_{n=1}^\infty z_n\otimes Ae_n\right)=(A\vee \mathbbm{1})^\ast\Omega=A^\ast\Omega=\sum_{n=1}^\infty z_n \otimes A^\ast e_n,\\
	F_0((\mathbbm{1}\vee A)\Omega)&=F_0\left(\sum_{m=1}^\infty Ae_m\otimes y_m\right)=(\mathbbm{1}\vee A)^\ast\Omega=\Omega A=\sum_{m=1}^\infty A^\ast e_m \otimes y_m.\\
	\end{aligned}$$
	By taking $A=e_i\otimes e_j$, we have that
	\begin{equation}
	\label{calculation3}
	S_0\left(z_i\otimes e_j\right)=S_0\left(\sum_{n=1}^\infty z_n\otimes Ae_n\right)=\sum_{n=1}^\infty z_n \otimes A^\ast e_n=z_j \otimes e_i,
	\end{equation}
	and, doing the analogous calculation,
	\begin{equation}
	\label{calculation3.1}
	F_0\left(e_j\otimes y_i\right)=F_0\left(\sum_{m=1}^\infty A e_m\otimes y_n\right)=\sum_{m=1}^\infty A^\ast e_m \otimes y_m=e_i\otimes y_j.
	\end{equation}
	
	Let $T:\hilbert \to \hilbert$ such that $T(e_i)=y_i$. Using Theorem \ref{eqciclicseparating}, one can check that $T$ is injective, but it is not surjective since $(y_i)_{i\in\mathbb{N}}\in\ell_2(\hilbert)$. It is also important to notice that $\Phi:\ell_2\to \hilbert$ given by $\Phi((\alpha_n)_{n\in\mathbb{N}})=\sum_{i=1}^\infty \alpha_i e_i$ is an isometric isomorphism. Hence $T$ is a $2$-compact operator (an extensive treatment of $p$-compact operators can be found in \cite{Silva2013}), since $T\Phi(B_1^{\ell_2(\mathbb{C})})=T(B_1^\hilbert)$ is a $2$-compact set. One can check that, since we are in the Hilbert space case, to be $2$-compact is equivalent to be a Hilbert-Schmidt operator and we will prefer the last denomination since it is the most spread one. Let also $T^{-1}:\Ran{T}\to \hilbert$ denote the inverse of $T$, $T=U|T|$ and $T=V|T^\ast|$, where $U$ and $V$ are unitary operators (since $\Ran{T}$ and $\Ran{T^\ast}$ are dense in $\hilbert$), be the polar decompositions of $T$ and $T^\ast$, respectively. Then,
	$$T^\ast=V\left|T^\ast\right|=VU|T|U^\ast \Rightarrow T=U|T|(VU)^\ast=T(VU)^\ast\Rightarrow T(\mathbbm{1}-(VU)^\ast)=0,
	$$	it follows from the fact that $T$ is injective that $(VU)^\ast=\mathbbm{1}$. Hence, $V=U^{-1}=U^\ast$.

	Notice that $T=\sum_{i=1}^\infty y_i\otimes e_i$. Hence $$T^\ast=\sum_{i=1}^\infty e_i\otimes y_i=\sum_{i=1}^\infty\sum_{n=1}^\infty a_{in} \ e_i\otimes e_n=\sum_{n=1}^\infty \left(\sum_{i=1}^\infty a_{in}^\ast e_i\right)\otimes e_n=\sum_{n=1}^\infty z_n \otimes e_n.$$
	
	Then
	\begin{equation}
	\begin{aligned}
	FS(T^\ast T e_i\otimes e_j)\quad &=\quad FS\left(\left(T^\ast \sum_{k=1}^\infty a_{ik} e_k\right)\otimes e_j\right)
	&=&\quad FS\left(\left(\sum_{k=1}^\infty a_{ik} z_k\right)\otimes e_j\right)\\
	&=\quad F\left(\sum_{k=1}^\infty a_{ik} S\left(z_k\otimes e_j\right)\right)
	&=&\quad F\left(\sum_{k=1}^\infty a_{ik}  \ z_j\otimes e_k\right)\\
	&=\quad F\left( z_j\otimes\left(\sum_{k=1}^\infty a_{ik} e_k\right)\right)
	&=&\quad F\left( z_j\otimes y_i\right)\\
	&=\quad F\left(\left(\sum_{l=1}^\infty a_{lj}^\ast e_l\right)\otimes y_i\right)
	&=&\quad \sum_{l=1}^\infty a_{lj}^\ast F\left( e_l \otimes y_i\right)\\
	&=\quad \sum_{l=1}^\infty a_{lj}^\ast \ e_i \otimes y_l
	&=&\quad e_i \otimes \left(\sum_{l=1}^\infty a_{lj}^\ast  T e_l\right)\\
	&=\quad e_i \otimes \left( T \sum_{l=1}^\infty a_{lj}^\ast e_l\right)
	&=&\quad e_i \otimes \left( T T^\ast e_j\right)\\
	\end{aligned}
	\end{equation}

	It is important to stress that we can commute $S$ or $F$ with the infinity sums in the steps above because $S$ and $F$ are closed operators and both sums are convergent.
	
	Hence, by linearity and closedness, for any $x\in \Dom{T^{-1} (T^\ast)^{-1}}$
	\begin{equation}\label{modoperator}
	\begin{aligned}
	\Delta_\Omega \left(x\otimes y\right)&=FS\left(x\otimes y\right)\\
	&=FS\left(T^\ast T [T^{-1} (T^\ast)^{-1} x] \otimes y\right)\\
	&=\left(T^{-1} (T^\ast)^{-1} x\right)\otimes \left(TT^\ast y\right)\\
	&=\left(|T|^{-2} x\right)\otimes \left(|T^\ast|^2 y\right)\\
	&=\left(|T|^{-2}\vee \mathbbm{1}\right)\left(\mathbbm{1}\vee|T^\ast|^2\right) x\otimes y
	\end{aligned}\end{equation}
	
	Since
	$$\begin{aligned}
	z_j \otimes e_i&=S(z_i\otimes e_j)\\
	&=J\Delta_\Omega^{\frac{1}{2}}(z_i\otimes e_j)\\
	&=J\left(|T|^{-1} z_i\right)\otimes \left(|T^\ast| e_j\right)\\
	&=J\left(U^\ast (T^\ast)^{-1} z_i\right)\otimes \left(V^\ast z_j\right)\\
	&=J\left(U^\ast e_i\right)\otimes \left(V^\ast z_j\right)\\
	&\Rightarrow J(e_i\otimes e_j)=U^\ast e_j\otimes Ue_i
	\end{aligned}$$
	
	In the light of the Tomita-Takesaki Theorem, it is interesting to compute the effect of of $JA\vee \mathbbm{1}J$ for an operator $A \in HS(\hilbert)$.
	\begin{equation}\label{modconj}\begin{aligned}
	J (A\vee\mathbbm{1})J x\otimes y&=J (A\vee\mathbbm{1}) U^\ast y\otimes Ux\\
	&=J \ U^\ast y\otimes AUx\\
	&=U^\ast AUx\otimes UU^\ast y\\
	&=U^\ast AUx\otimes y\\
	&=U^\ast AUx\otimes UU^\ast y\\
	&=(\mathbbm{1}\vee U^\ast AU)x\otimes y,
	\end{aligned}\end{equation}
	thus $J(A\vee\mathbbm{1})J=(\mathbbm{1}\vee U^\ast AU)\in \calgebra_r$, as it should be.
	
	\subsection{Hamiltonian}\label{H}
	
	The modular Hamiltonian $H$ is defined in analogy with the case where the system is described by a density matrix by $\Delta_\Omega=e^{H}$. Here it is important to remember that there is a difference between physicists' and mathematicians' convention: remember that the modular condition, which is satisfied by the modular automorphism group, is equivalent to the $(-1)$-KMS condition.
	
	\begin{equation}\begin{aligned}\label{modularhamiltonian}
	H&\defeq\log{\Delta_\Omega}\\
	&=\log{\left(|T|^{-2}\vee \mathbbm{1}\right)\left(\mathbbm{1}\vee|T^\ast|^2\right)}\\
	&=-2\log{\left(|T|\vee \mathbbm{1}\right)}+2\log{\left(\mathbbm{1}\vee|T^\ast|\right)}\\
	&=2\left(\log{|T|}\vee \mathbbm{1}\right)-2\left(\mathbbm{1}\vee\log{|T^\ast|}\right)\\
	\end{aligned}\end{equation}
	
	On the other hand, it is well known that the modular Hamiltonian defined in equation \eqref{modularhamiltonian} is related to the Hamiltonian by $H=h\vee\mathbbm{1}-Jh\vee\mathbbm{1}J$. Using equation \eqref{modconj} we get
	\begin{equation}h=-2\log{|T|}.\end{equation}
	
	Here we finish mentioning that we have characterized all the cyclic and separating vectors in the algebra and, therefore, all Hamiltonians that reach thermal equilibrium.
	
	A curious consequence is that if $\hilbert$ is infinite dimensional $h$ cannot be bounded bellow (or above in physicists' convention) in order to exist the thermal equilibrium state, since $0\in \sigma(|T|)$. On the other hand, by the very same reason, the Hamiltonian must always be bounded from above (from bellow in physicists convention).
	
	The fact that $T$ is a Hilbert-Schmidt operator imposes strong constraints to the possible Hamiltonians. In order to better understand this, we need to diagonalize $|T|$. Since $T\in HS(\hilbert)$ is injective, we can write the spectral decomposition of $|T|$ as $|T|=\sum_{i=1}^{\infty} \lambda_i e_i\otimes e_i$, where each $\lambda_i$ is strictly positive, $\lambda_{i+1}<\lambda_i$, and $\sum_{i=1}^\infty \lambda_i^2=\|T\|^2_2$. Then,
	$\displaystyle h=-2\log{|T|}=-2\sum_{i=1}^\infty \log{\lambda_i}e_i\otimes e_i$, from where it follows that $h$ must have a $\infty$-divergent discrete spectrum satisfying \begin{equation}\label{condham}\sum_{i=1}^\infty e^{-h_i}=\sum_{i=1}^\infty e^{2\log\lambda_i}=\sum_{i=1}^\infty \lambda_i^2=\|T\|_2^2<\infty.\end{equation}
	
	The converse is also true: if the Hamiltonian $h$ is a self-adjoint operator such that $e^{\frac{1}{2}h}$ is a Hilbert-Schmidt operator, then it admits a KMS state.

Let us obtain the KMS state to the Landau levels problem doing the reverse calculation. As we discussed in Section \ref{LandauLevels}, the Landau levels problem is fully characterized by having two quantum numbers - let's say $m$ and $n$ with $n\geq 0$ and $m\geq -n$ - associated to two (unbounded) commuting operators $a$ and $b$ affiliated to the algebra and satisfying equations \eqref{laddereigen} and related to the Hamiltonian and angular momentum by equations \eqref{eigenvalues}.
	
The association $|n,m\rangle=e_{n+1}\otimes e_{m+n+1}$ is now natural.
	
We have obtained that $$h=-2\log|T|\Rightarrow |T|\ |n,m\rangle=e^{-\frac{h}{2}}|n,m\rangle=e^{-\frac{1}{2}\left(n+\frac{1}{2}\right)} |n,m\rangle,$$
which warranties that $|T|$ is not in the $2$-Schatten class due to its degenerescence in $m$.
	
The conclusion is that there isn't a faithful normal KMS state for the Landau levels in this description, which contrasts with \cite{bagarello2010}. The difference is that in \cite[Section III]{bagarello2010} it is considered that the Hilbert space has finite dimension, from which follows that linear operators are Hilbert-Schmidt operators. The conclusion is that there exists a faithful normal KMS state in this construction if, and only if, the Hilbert space has finite dimension. In \cite[Section IV]{bagarello2010} the Hilbert space considered is $L_2(\mathbb{R})$, which is not finite dimensional. In addition, the ``density matrix'' seems to should be trace class as an operator on $\tilde{\hilbert}=HS(\hilbert)$ which would be true if, and only if, we restrict our analysis to an algebra with finite degeneracy in energy which is in complete agreement with the construction of the present work in the infinite dimensional case.

We are going to analyse the limit of the finite box (which happens to have a description in a finite dimensional subspace) in the next section.
	
Hence, only few options remain for the thermodynamical equilibrium state in the thermodynamical limit: (i) it does not exist; (ii) it is not faithful; (iii) it is not normal; (iv) it is not a KMS state at all.

In order to better analyse the physical situation we need to take the limit of the system confined in a bounded region. The next section is devoted to that.

\section{Degeneracy in the finite (circular) box}\label{sec:finitebox}
	
	There is standard literature about the degeneracy of the Landau levels problem, but here we decide to present a different approach. Since we are in the symmetric gauge it makes more sense to impose that $x^2+y^2\leq R^2$.
	
	Since the operators $a$ and $b$, and their respectively adjoints, are defined by equations \eqref{ladderop}, one can check that $x^2+y^2=aa^\dagger+a^\dagger a+2\left(a b^\dagger + a^\dagger b\right) + b b^\dagger+b^\dagger b$ and that
	$$\begin{aligned}
	\left(x^2+y^2\right)\sum &a_{n,m}|n,m\rangle=\sum \left(4n+2m+2\right)a_{n,m}|n,m\rangle+
	\\
	&\hfill+2\left(\sum \sqrt{n}\sqrt{n+m+1}\ a_{n,m}  |n-1,m+2\rangle+\sum \sqrt{n+1}\sqrt{n+m}\ a_{n,m}|n+1,m-2\rangle\right)
	\end{aligned},$$
hence, in order to have an eigenvector with eigenvalue $\lambda$, we must have
\begin{equation}
\label{coeff}
2\sqrt{n+1}\sqrt{n+m}\ a_{n+1,m-2}+\left(4n+2m+2-\lambda\right)a_{n,m}+2\sqrt{n}\sqrt{n+m+1}\ a_{n-1,m+2}=0,
\end{equation}
where we are considering that $a_{n,m}=0$ for $n<0$ and $m<-n$.

$$\begin{tikzpicture}
\matrix[matrix of math nodes,row sep=15pt,column sep = 15pt] (m){
	&\ &\ &\ &\ &0 &0 &0 &0 &0 &\cdots\\
	&\ &\ &\ &0 &a_{0,0} &a_{0,1} &a_{0,2} &a_{0,3} &a_{0,4} & \cdots\\
	&\ &k=0\rightarrow \hspace{-0.7cm}&0 &a_{1,-1} &a_{1,0} &a_{1,1} &a_{1,2} &a_{1,3} &a_{1,4} & \cdots\\
	&k=1\rightarrow \hspace{-0.7cm} &0 &a_{2,-2} &a_{2,-1} &a_{2,0} &a_{2,1} &a_{2,2} &a_{2,3} &a_{2,4} & \cdots\\
	k=2\rightarrow \hspace{-0.6cm}&0 &a_{3,-3} &a_{3,-2} &a_{3,-1} &a_{3,0} &a_{3,1} &a_{3,2} &a_{3,3} &a_{3,4} & \dots\\
	\udots&\vdots &\vdots &\vdots &\vdots &\vdots &\vdots &\vdots &\vdots &\vdots & \ddots\\
};

\draw (m-3-4.north east) -- (m-3-4.south east) -- (m-3-4.south west) -- (m-3-4.north west) -- (m-3-4.north east);
\draw (m-3-4.north east) -- (m-2-6.south west);
\draw (m-2-6.north east) -- (m-2-6.south east) -- (m-2-6.south west) -- (m-2-6.north west) -- (m-2-6.north east);
\draw (m-2-6.north east) -- (m-1-8.south west);
\draw (m-1-8.north east) -- (m-1-8.south east) -- (m-1-8.south west) -- (m-1-8.north west) -- (m-1-8.north east);

\draw (m-4-3.north east) -- (m-4-3.south east) -- (m-4-3.south west) -- (m-4-3.north west) -- (m-4-3.north east);
\draw (m-4-3.north east) -- (m-3-5.south west);
\draw (m-3-5.north east) -- (m-3-5.south east) -- (m-3-5.south west) -- (m-3-5.north west) -- (m-3-5.north east);
\draw (m-3-5.north east) -- (m-2-7.south west);
\draw (m-2-7.north east) -- (m-2-7.south east) -- (m-2-7.south west) -- (m-2-7.north west) -- (m-2-7.north east);
\draw (m-2-7.north east) -- (m-1-9.south west);
\draw (m-1-9.north east) -- (m-1-9.south east) -- (m-1-9.south west) -- (m-1-9.north west) -- (m-1-9.north east);

\draw (m-5-2.north east) -- (m-5-2.south east) -- (m-5-2.south west) -- (m-5-2.north west) -- (m-5-2.north east);
\draw (m-5-2.north east) -- (m-4-4.south west);
\draw (m-4-4.north east) -- (m-4-4.south east) -- (m-4-4.south west) -- (m-4-4.north west) -- (m-4-4.north east);
\draw (m-4-4.north east) -- (m-3-6.south west);
\draw (m-3-6.north east) -- (m-3-6.south east) -- (m-3-6.south west) -- (m-3-6.north west) -- (m-3-6.north east);
\draw (m-3-6.north east) -- (m-2-8.south west);
\draw (m-2-8.north east) -- (m-2-8.south east) -- (m-2-8.south west) -- (m-2-8.north west) -- (m-2-8.north east);
\draw (m-2-8.north east) -- (m-1-10.south west);
\draw (m-1-10.north east) -- (m-1-10.south east) -- (m-1-10.south west) -- (m-1-10.north west) -- (m-1-10.north east);
\end{tikzpicture}$$

Notice that we highlight the elements related by equation \eqref{coeff} and that there exists an equation for each three consecutive coefficients. Hence for each set of highlighted coefficients we have the same number of elements as we have equations. Thus, all the coefficients are determined by solving a (finite) tridiagonal system of linear equations. Since these systems of equations are always homogeneous and we are looking for a non-null solution, we can require that the determinant of the coefficients must vanish. We stress that this requirement implies that $\lambda$ (which is the eigenvalue of the operator $x^2+y^2$) must be a root of a polynomial of degree $n$, which always exists. Finally, since $x^2+y^2$ is a symmetric operator $\lambda$ must be a real number. Notice that if we set as zero all coefficients but the ones in a particular diagonal line, we obtain an eigenvector.

The tridiagonal matrix corresponding to the system starting in the $k$-line, $0\leq k$, is

$$A_k\hspace{-0.2em}=\hspace{-0.4em}\begin{tikzpicture}[every left delimiter/.style={xshift=0.9em},
every right delimiter/.style={xshift=-0.9em}, baseline=(current bounding box.center)]
\matrix (m) [matrix of math nodes,nodes in empty cells,right delimiter={]},left delimiter={[} ]{
	 (2k+2-\lambda) & 2\sqrt{k}\sqrt{1} &0 & & & &0\\
	2\sqrt{k}\sqrt{1}& (2k+2-\lambda)&2\sqrt{k-1}\sqrt{2}& & & &\\
	0&2\sqrt{k-1}\sqrt{2}& (2k+2-\lambda) &2\sqrt{k-2}\sqrt{3}& & &\\
	 & &2\sqrt{k-2}\sqrt{3}&(2k+2-\lambda) &2\sqrt{k-3}\sqrt{4}& &\\
	 & & & & & &\\
	 & & & & & &0\\
	 & & & & & &2\sqrt{1}\sqrt{k}\\
	0& & & & 0&2\sqrt{1}\sqrt{k}&(2k+2-\lambda)\\
};
\draw[loosely dotted] (m-3-1)-- (m-8-1);
\draw[loosely dotted] (m-3-1)-- (m-8-5);
\draw[loosely dotted] (m-8-1)-- (m-8-5);
\draw[loosely dotted] (m-4-3)-- (m-8-6);
\draw[loosely dotted] (m-4-4)-- (m-8-7);
\draw[loosely dotted] (m-4-5)-- (m-7-7);
\draw[loosely dotted] (m-1-3)-- (m-1-7);
\draw[loosely dotted] (m-1-7)-- (m-6-7);
\draw[loosely dotted] (m-1-3)-- (m-6-7);
\label{figureeigenvectors}
\end{tikzpicture}$$

We recognise that this is equivalent to finding the eigenvalues of the matrix $B_k=A_k+\lambda \mathbbm{1}_k$, which is a Jacobi matrix (which is not surprising, since probably it is due to some connection with orthogonal polynomials), hence we know that this matrix has simple eigenvalues.

In order to find its eigenvalues, notice that
if we have $2k+2-\lambda=2k$ we can cancel the first term in the lower-diagonal by subtracting from it the first line times $\frac{1}{\sqrt{k}}$ and the last upper-diagonal term doing the same using the second to last row. The resulting matrix will be
$$\tilde{B}_k=\begin{bmatrix}
2k&2\sqrt{k}&0&\cdots & 0\\
0&2k-2&2\sqrt{k-1}\sqrt{2}&\ddots&\vdots\\
0&2\sqrt{k-1}\sqrt{2}&2k&\ddots&0\\
\vdots&\ddots&\ddots&2k-2 & 0\\
0&\cdots& 0 &2\sqrt{k} &2k
\end{bmatrix}.$$

Notice that $\det B_k=(2k)^2 \det \left(\tilde{B}_k\right)_{\{1,k+1\}}$ were $\left(\tilde{B}_k\right)_{\{1,k+1\}}$ is the principal submatrix obtained by omitting the first and last rows and columns. We can proceed inductively by multiplying the replacing the $j$-row by its sum with $\frac{\sqrt{j}}{\sqrt{k-j+1}}$ and do the same with the $(k-j+1)$-row until it remains only a one-row or two-row square central block unchanged, depending on $k$ being even or odd, respectively. One can check that the central blocks are just the null element, in case $k$ is even, and $\begin{bmatrix}
k+1 & k+1\\
k+1& k+1\\
\end{bmatrix}$ in case k is odd, and in both the determinant vanish.

More generally by using the similar procedure, one can check that the determinant of the matrix $A_k$ satisfies $2k+2-\lambda = 2k-4j$ for $0\leq j \leq k$. We just would like to add that we explicitly used that $2k-2j$ is never zero for $1\leq j\leq {k+1}{2}$, but will not be true in the procedure to $2k-4j$. In this case, one can check that the vanishing term just occurs in the $k-j$-th row and the determinant is equal to the product of the diagonal terms, omitting the vanishing one, and the determinant of the remaining $(j+1)$-size block, for which our elimination can be performed again.

The conclusion is that the eigenvalues of $A_k$ are $2+4j$ for $0\leq j\leq k$. This shows us not only that the operator $x^2+y^2$ is an unbounded operator (which was quite obvious from the beginning), but also gives us some important information about the spectral projection of this operator. In particular, to ensure that the spectrum of $x^2+y^2$ is limited by a quantity $R^2$, we must restrict our Hilbert space in such a way that it does not contain infinitely many diagonal lines. We can use that the largest eigenvalue of $x^2+y^2$ is $2+4k$ to obtain that in order of the spectrum $\sigma(x^2+y^2)\subset[0,R^2]$ is necessary and sufficient that $\lambda_{max} =2+4k\leq R^2\Rightarrow k\leq \frac{R^2-2}{4}$. Here is important to remember that $k$ refers to the diagonal line in figure \ref{figureeigenvectors} that starts at the $k$-th energy level. Its also interesting to note that if we ``cut'' the Hilbert space in the $k$-th diagonal line (which is equivalent to consider the space $E^{x^2+y^2}_{(-\infty, k]}\hilbert$), the $n$-th energy level, with $0\leq n\leq k$ has degeneracy equals $k-n+1$. Hence, the total number of states considering this cut is $\frac{(k+2)(k+1)}{2}=\frac{(R^2+6)(R^2+2)}{8}$, which grows faster than the area in the limit $k\to \infty$.

It is well known that the occupation number of a system at the energy level $i$ with energy $E_i$ and degeneracy $d_i$ is equal to
$N_n^k=\frac{d_i}{e^{\beta E_i}+1}=\frac{k-n+1}{e^{\beta \left(n+\frac{1}{2}\right)+1}}$ in the fermionic case (which is the case if we want restrict our discussion to electrons) and $N_n^k=\frac{d_i}{e^{\beta E_i}-1}=\frac{k-n+1}{e^{\beta \left(n+\frac{1}{2}\right)}-1}$ in the bosonic case.

Notice that, for any fixes $i$ and $j$, in the limit $k\to \infty$, $\frac{N_i^\infty}{N_j^\infty}=\frac{1\pm e^{\beta i}}{1\pm e^{\beta j}}$, which is exactly the ratio obtained in the usual bosonic and fermionic cases and has a strong connection with the condition $\sum_{n=0}^k N_n^k=N$. This leads to the conclusion that the probability to find a particle in a state with energy $E_n$, in the thermodynamic limit, is $P(E_n)=\frac{e^{-\beta E_n}}{Z}$, where $Z=\sum_{n=1}^\infty e^{-\beta E_n}$. On the other hand, the probability of finding a particle in the specific state $|n,m\rangle$, in the thermodynamic limit, is zero.

Such behaviour is not surprising and the reason for that is the non-$\sigma$-additivity of the probability limit. Let's give a simple example that illustrate the situation: in a box with $k$ identical balls numbered from $1$ to $k$, what is the probability of draw a ball with a specific number? It is obviously $P(j)=\frac{1}{k}$ and, in the limit $k\to \infty$, it results in $P(j)=0$. On the other hand, the probability of draw any number or an even number are $P(\textrm{any number})=1\to 1$ and $P(\textrm{even number})=\frac{\lfloor \frac{k-1}{2}\rfloor}{k}\to \frac{1}{2}$. Notice that, in the limit $k\to\infty$ we can define the probability of the number in a draw be part of a set $A\subset{N}$ is $P(A)=\lim_{k\to\infty}\frac{\#A\cap \{1,2,\cdots, k\}}{k}$, which is additive and satisfies every probability axiom but $\sigma$-additivity. These ideas are also connected with the difference between improbable and impossible. 

Here it is interesting to note that, in the usual situation in statistical mechanics, $e^{-\beta H}$ is a trace-class operator, therefore compact. Hence, thanks to Gröthendiek's theorem, its eigenvalues in decreasing order must form a sequence that converges to zero, from which we conclude that infinite degeneracy broken when we restrict our system to a finite box is a quite general behaviour in quantum statistical mechanics. However, no general behaviour of the occupation number is expected in the thermodynamical limit, since it is the limit of the ratio of the occupation numbers that seems to play the relevant role in this situation and it may have very distinctive behaviour. A counter intuitive physical system would be one such that the degeneracy of the ground level (or any fixed level) grows so fast that $\lim_{k\to \infty} \frac{N_0^k}{\sum_{n=0}^k N_i^k}\to 1$. This situation should not be confused with some kind of Bose-Einstein condensation, since, although all the particles are in the same energy level, it is expected to have just few particles in each quantum state. However the author is not aware of any physical system with this property.

Back to the options listed in Section \ref{classification}, notice that the expected thermodynamical limit from the physical view point is not expected to be neither (ii) faithful nor (iii) normal. If only fidelity was missing, it would not be a problem, since we could just take the quotient by the kernel of the state. Notice that this would means that all information about the quantum number related to the angular momentum would be lost and, in this situation, only projections related to the energy quantum number would exist, hence we would have a KMS.
Nevertheless, this seems not to describe the complete physical situation and just lost all the angular momentum quantum number is not reasonable, since it is valid to ask the probability finding the state with $m$ being even. It also call our attention to the important fact that been a KMS state strongly depends on the algebra we are considering.

\section{Conclusions}
	
The main results showed in the present paper are: 
\begin{enumerate}[(i)]
	\item We classified all cyclic and separating vectors of a von Neumann algebra that seems suitable for the description of system with infinity degeneracy.
	 
	\item We showed that the Hamiltonian of the Landau levels and more general situations of system with infinitely degeneracy in energy does not admit faithful normal KMS states.
	
	\item We analysed the thermodynamical limit of the system restricted to a finite box, since there is no infinity degeneracy, and showed that this limit in general is neither faithful nor normal.
\end{enumerate}

Is the present opinion of the author that the restriction of a KMS states been normal in $W^\ast$-dynamical systems should be replaced for quantum systems with infinite degeneracy in energy by the state be a normal KMS state only on the von Neumann algebra generated by the projections on each energy level. Moreover, this motivates the study of singular traces, e.g. Dixmier traces, as physically motivated ones.

	\bibliographystyle{ieeetr}
	\bibliography{references}

\begin{thebibliography}{10}

\bibitem{Bratteli1}
O.~Bratteli and D.~W.~R. Robinson, {\em Operator Algebras and Quantum
  Statistical Mechanics 1}.
\newblock Berlin: Springer, second~ed., 1987.

\bibitem{Araki74}
H.~Araki, ``Some properties of modular conjugation operator of von {N}eumann
  algebras and a non-commutative {R}adon-{N}ikodym theorem with a chain rule,''
  {\em Pacific Journal of Mathematics}, vol.~50, no.~2, pp.~309--354, 1974.

\bibitem{RCS18}
R.~{Correa da Silva}, ``Lecture notes on noncommutative {$L_p$}-spaces,'' Mar.
  2018.

\bibitem{Takesaki2002}
M.~Takesaki, {\em Theory of operator algebras I}, vol.~124 of {\em
  Encyclopaedia of Mathematical Sciences}.
\newblock Springer-Verlag Berlin Heidelberg, 2002.

\bibitem{KR83}
R.~V. Kadison and J.~R. Ringrose, {\em Fundamentals of the {Theory} of
  {Operator} {Algebras}, {Part} {I}}.
\newblock New York: Academic Press, 1983.

\bibitem{Bratteli2}
O.~Bratteli and D.~W.~R. Robinson, {\em Operator Algebras and Quantum
  Statistical Mechanics 2}.
\newblock Berlin: Springer, second~ed., 1997.

\bibitem{KR86}
R.~V. Kadison and J.~R. Ringrose, {\em Fundamentals of the {Theory} of
  {Operator} {Algebras}, {Part} {II}}.
\newblock New York: Academic Press, 1986.

\bibitem{haag67}
R.~Haag, N.~M. Hugenholtz, and M.~Winnink, ``On the equilibrium states in
  quantum statistical mechanics,'' {\em Communications in Mathematical
  Physics}, vol.~5, no.~3, pp.~215--236, 1967.

\bibitem{bagarello2010}
S.~T. Ali, F.~Bagarello, and G.~Honnouvo, ``Modular structures on trace class
  operators and applications to landau levels,'' {\em Journal of Physics A:
  Mathematical and Theoretical}, vol.~43, no.~10, p.~105202, 2010.

\bibitem{Silva2013}
R.~{Correa da Silva}, ``Operadores p-compactos e a propriedade de
  p-aproxima\c{c}\~ao,'' Master's thesis, Universidade de S\~ao Paulo, 2013.
\newblock
  http://www.teses.usp.br/teses/disponiveis/45/45131/tde-12112013-232741/en.php.

\end{thebibliography}

\end{document}